\documentclass[conference]{IEEEtran}
\IEEEoverridecommandlockouts
\usepackage{cite}
\usepackage{amsmath,amssymb,amsfonts}
\usepackage{algorithmic}
\usepackage{algorithm}
\usepackage{graphicx}
\usepackage{textcomp}
\usepackage{xcolor}
\usepackage{booktabs}
\usepackage{multirow}
\usepackage{array}
\newtheorem{definition}{Definition}
\newtheorem{assumption}{Assumption}
\newtheorem{proposition}{Proposition}
\newtheorem{corollary}{Corollary}
\newtheorem{lemma}{Lemma}
\def\BibTeX{{\rm B\kern-.05em{\sc i\kern-.025em b}\kern-.08em
    T\kern-.1667em\lower.7ex\hbox{E}\kern-.125emX}}
\usepackage[utf8]{inputenc}
\usepackage{newunicodechar}
\usepackage{url}
\newunicodechar{∼}{\textasciitilde}
\begin{document}

\title{Covering the Unseen: Information Demand Coverage Optimization for Retrieval-Augmented Generation}

\author{
\IEEEauthorblockN{
Bingxue Zhang\textsuperscript{1},
Jianying Jia\textsuperscript{1},
Feida Zhu\textsuperscript{2,*}
}
\IEEEauthorblockA{
\textsuperscript{1}\textit{University of Shanghai for Science and Technology}, Shanghai, China\\
\textsuperscript{2}\textit{Singapore Management University}, Singapore\\
Emails: zhangbingxue@usst.edu.cn, jiajianying@st.usst.edu.cn, fdzhu@smu.edu.sg\\
\textsuperscript{*}Corresponding author
}
}
\maketitle

\begin{abstract}
RAG typically models context selection as a chunk-wise ranking problem, assuming that a query's information need can be captured by a single embedding vector. However, this assumption often fails for complex queries such as multi-hop reasoning and ambiguous QA, causing top-$k$ selections to concentrate on a single semantic dimension while leaving critical sub-questions unaddressed. To address this limitation, we propose GeoRAG, which reformulates context selection as \emph{Information Demand Coverage Optimization}. GeoRAG constructs a multi-dimensional Information Demand Proxy distribution via two-stage diverse sub-query generation and reverse-validation quality weighting, and optimizes selection by minimizing the entropic Sinkhorn--Wasserstein distance between the demand distribution and the coverage measure of the selected set. We define a demand-weighted facility-location objective that is monotone submodular and admits the classical greedy $1-\frac{1}{e}$ approximation; GeoRAG implements a Sinkhorn-based marginal-gain surrogate that approximates this coverage objective and is validated empirically through ablations, enabling balanced coverage across multi-dimensional demands in an unsupervised, training-free, and retrieval-agnostic manner. We further show that query-proximity-monotone scorers based on a single-point need representation cannot adequately cover multi-modal demands, revealing a structural limitation of conventional ranking-based selection. Experiments on six open-domain QA benchmarks demonstrate consistent EM improvements of $+6.5$ to $+7.5$ over top-$k$ truncation and superiority over strong baselines including MMR, DPP, BGE-Reranker, SMART-RAG, and AdaGReS. The greatest gains reach $+9.7$ on both HotpotQA and ASQA. GeoRAG also substantially improves demand-dimension coverage, and the performance gains remain stable across different context budgets and sub-query generators, confirming the effectiveness of the coverage-based formulation.
\end{abstract}

\begin{IEEEkeywords}
RAG, Information Demand Coverage Optimization, Sinkhorn--Wasserstein
\end{IEEEkeywords}

\section{Introduction}
\label{sec:introduction}

\subsection{Retrieval-Augmented Generation and the Context Selection Problem}

Retrieval-Augmented Generation has emerged as a core paradigm for enhancing the factual accuracy of Large Language Models by retrieving relevant paragraphs from external knowledge bases at inference time to inject into the generator's context \cite{lewis2020rag}. Its standard pipeline consists of three steps: a dense retriever recalls $K$ candidate chunks from a corpus of millions of documents (typically $K=200$). Then $k$ chunks are selected from these candidates to form the context window (typically $k=5$). Finally, this window is fed into the generator to produce an answer. Although the retrieval stage, that is, how to efficiently and accurately recall candidate chunks, has been extensively investigated, both at the model level \cite{karpukhin2020dpr, xiong2021ance} and at the systems level through purpose-built vector data management and approximate nearest-neighbor indexing \cite{wang2021milvus, pan2024vectorsurvey}, with recent database research continuing to push range-filtered search \cite{zuo2024serf}, adaptive index construction \cite{mageirakos2025crackivf}, and automated tuning of vector data management systems \cite{yang2024vdtuner}, the second step -- selecting which $k$ chunks to keep from the $K$ candidates, called \emph{context selection} -- has long been treated as a simple ranking truncation, with its underlying design assumptions left unquestioned.

\subsection{Structural Limitations of Existing Methods}

Existing context selection methods, whether chunk-wise ranking or set diversity approaches, share an unexamined assumption: the information need of a query $Q$ can be fully expressed by a single embedding vector $\text{emb}(Q)$. Chunk-wise ranking methods, such as Cosine top-$k$, BGE-Reranker \cite{xiao2023bge}, and RankRAG \cite{yu2024rankrag}, independently score each candidate chunk as $f(c_i, Q)$. This essentially measures the distance between the candidate and the single point $\text{emb}(Q)$, making the selection process entirely blind to the dimensions of the information that have already been covered by the selected set. Set diversity methods, such as MMR \cite{carbonell1998mmr}, DPP \cite{kulesza2012dpp}, SMART-RAG \cite{li2024smartrag} and AdaGReS \cite{peng2025adagres}, introduce set awareness; however, their direction of diversity is determined isotropically by the geometric structure of the embedding space, equivalent to assuming that the query needs radiate uniformly from $\text{emb}(Q)$, regardless of the true semantic dimensional structure of the query. The common blind spot in both types of method is that none can answer the question of which information dimension of the query is under-covered by the current selected set $S$, and which gap the next chunk should prioritize filling. A direct consequence of this structural flaw is that the selected chunks are highly redundant: measured on the Dense retriever over the NQ and HotpotQA development sets, the average pairwise cosine similarity among the top-5 selected candidates is $0.82$ (mean over queries). Such massive semantic redundancy occupies the limited context window, leaving critical sub-questions perpetually unanswered.

\subsection{Typical Scenarios of Single-Point Need Assumption Failure}

Consider the multi-hop question: ``Who was the vice president of the president who signed the Voting Rights Act?'' The information need for this query exhibits a bimodal structure. Peak A, pointing to the identity of who signed the act and thus to Johnson, and Peak B, pointing to the vice president of that president and thus to Humphrey, are semantically independent. However, $\text{emb}(Q)$ lies closer to Peak A in the embedding space. Consequently, Cosine top-$k$ allocates all five slots to paragraphs related to Peak A. Even if critical chunks covering Peak B exist in the corpus, they are systematically ignored due to their greater distance from $\text{emb}(Q)$.

Crucially, this failure is one of \emph{selection}, not of \emph{recall}. Across the six retrievers we study, Recall@200 already ranges from $88.7\%$ to $93.8\%$: the evidence covering Peak B is almost always present in the candidate pool, yet it is discarded at the selection step. Moreover, across the six retrievers we study, the magnitude of GeoRAG's gain shows no positive association with retrieval quality (Spearman $\rho = -0.23$ over $n=6$, not significant; we treat this as suggestive rather than conclusive given the small sample, Section~\ref{sec:experiments}), consistent with the bottleneck lying in how the $k$ chunks are chosen from the pool rather than in how the pool is recalled. The root cause is therefore a structural flaw in the context selection objective: a single-point need representation cannot capture the multidimensional information structure of the query, as illustrated in Figure~\ref{Motivation} and made precise in Proposition~\ref{prop:noncover}.

\begin{figure}[htbp]
\centerline{\includegraphics[width=\columnwidth]{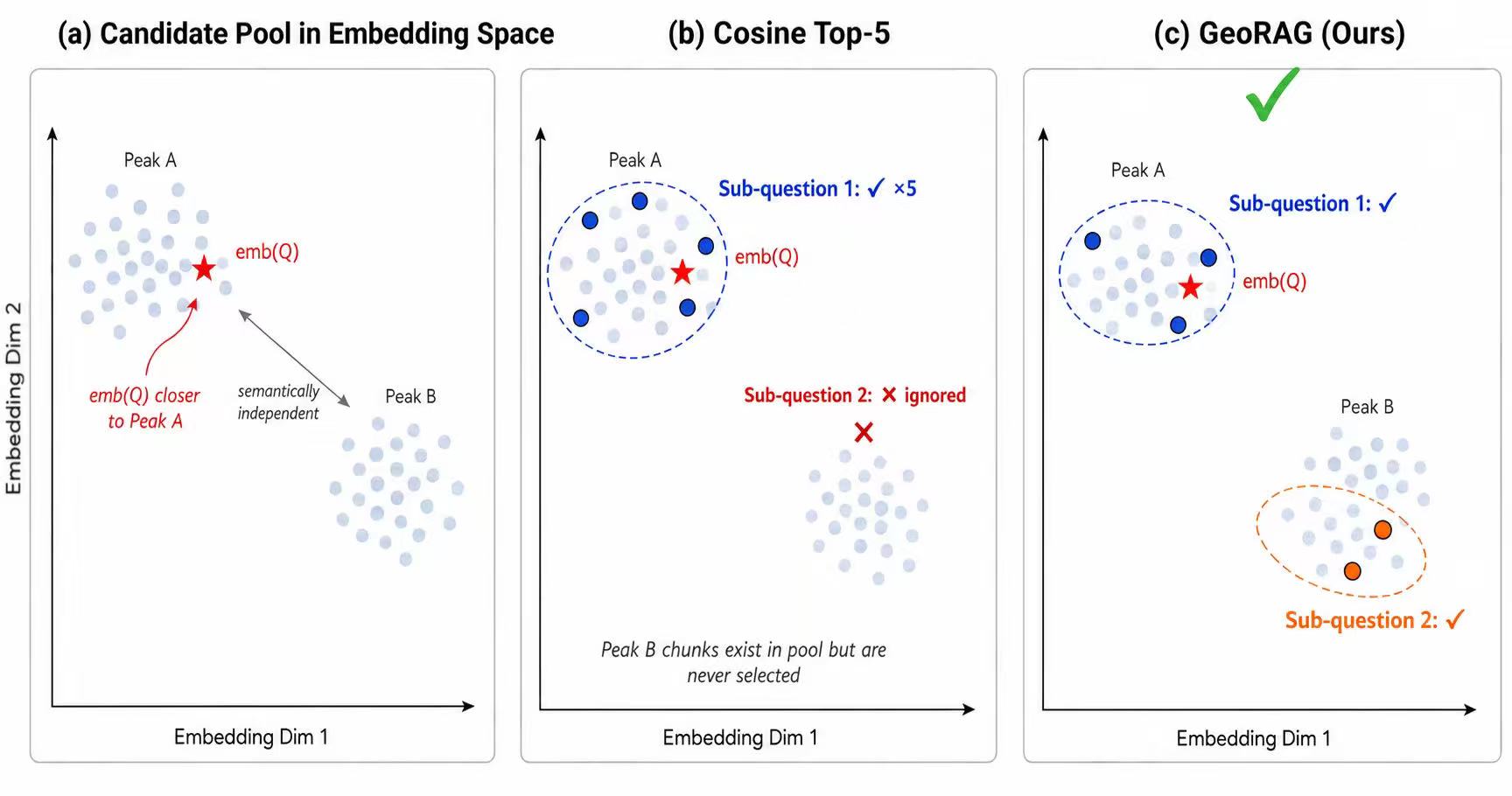}}
\caption{Motivation: single-point $\text{emb}(Q)$ (left) leads Cosine top-5 to crowd all
slots around Peak~A (center), while GeoRAG balances coverage across both demand peaks (right).}
\label{Motivation}
\end{figure}

\subsection{Method Overview and Main Contributions}

We reformulate context selection from a chunk-wise ranking problem into an Information Demand Coverage Optimization problem. Specifically, we construct a multi-dimensional Information Demand Proxy distribution $\mathcal{P}_Q$ and maximize its coverage by the selected set, realized by minimizing the Sinkhorn--Wasserstein distance between the demand-anchored coverage measure $\mathcal{E}_S$ and $\mathcal{P}_Q$.

Our formulation introduces two orthogonal axes. \textbf{Axis A} models \emph{where the demand lies} through $\mathcal{P}_Q$, specifying which semantic dimensions should be covered. \textbf{Axis B} models \emph{what remains uncovered} through set-aware selection, guiding where the next chunk should be allocated. Chunk-wise rankers capture neither axis, while diversity-based methods such as MMR and DPP provide only an undirected form of Axis~B and lack Axis~A entirely. The experiments show that Axis~A and Axis~B independently contribute $+2.7$ and $+3.4$ EM on HotpotQA, demonstrating complementary gains.

Optimizing this objective faces three challenges: (P1) $\mathcal{P}_Q$ is a latent variable without supervision; (P2) LLM-generated sub-queries vary in quality, with roughly 10--20\% exhibiting semantic drift or redundancy; and (P3) exact optimization is computationally intractable due to the combinatorial search space $\binom{200}{5}\approx 2.5\times10^9$ and the $O(n^3)$ complexity of Optimal Transport.

To address these challenges, we propose GeoRAG, an unsupervised, training-free RAG context selection framework, with the following three contributions.

\noindent\textbf{Contribution 1.} We reformulate the selection of the RAG context as the optimization of information demand coverage and \emph{prove} (Proposition~\ref{prop:noncover}) that any query-proximity-monotone selector is structurally incapable of covering a bimodal demand once its two peaks separate in query-proximity (Assumption~\ref{ass:bimodal}), independently of the scorer's functional form or model capacity. This establishes the limitation of single-point need representations as a representational barrier rather than an engineering deficiency; for cross-encoder rerankers, which fall outside this class, we corroborate the same failure empirically via BGE-Reranker (Section~\ref{sec:experiments}).

\noindent\textbf{Contribution 2.} We propose GeoRAG, an unsupervised, training-free framework that resolves the three core challenges through: (P1)~two-stage diverse sub-query generation to approximate the multi-dimensional demand structure $\mathcal{P}_Q$ (Axis~A); (P2)~reverse-validation quality weighting to suppress drifted sub-queries; and (P3)~a greedy facility-location coverage objective that is monotone submodular (with the classical $(1-1/e)$ guarantee for its exact marginal gain), realized in deployment by a Sinkhorn-based marginal-gain surrogate for set-aware, gap-directed selection (Axis~B).

\noindent\textbf{Contribution 3.} Across six QA benchmarks and six retrieval methods, GeoRAG achieves consistent average EM gains of $+6.5$--$+7.5$ over the top-$k$ truncation and outperforms all selection baselines. Ablations confirm that Axis~A and Axis~B are independently irreplaceable, contributing $+2.7$ and $+3.4$ EM respectively on HotpotQA.

\section{Related Work}
\label{sec:related}

\subsection{Chunk-wise Reranking Methods}

Modeling context selection as a chunk-wise ranking problem is the most mainstream paradigm. Early work used cosine similarity top-$k$ truncation as a baseline; recent research has shifted toward stronger reranking models. BGE-Reranker \cite{xiao2023bge} jointly encodes concatenated query-chunk sequences using a cross-encoder, achieving a strong baseline in multilingual scenarios. RankRAG \cite{yu2024rankrag} unifies context ranking and answer generation, surpassing GPT-4 on nine benchmarks through instruction tuning, representing the current state-of-the-art in segment reranking. InfoGain-RAG \cite{wang2025infogain} estimates the utility of each chunk through information gain-based reranking and filtering, but as with all chunk-wise methods it does not model multi-evidence interactions within the selected set. The common bottleneck of these methods lies in optimizing $f(c_i, Q)$ within the chunk-wise paradigm. The selection process remains entirely blind to the coverage status of the selected set. We make this rigorous for the query-proximity-monotone subclass in Proposition~\ref{prop:noncover}; for cross-encoder rerankers, whose scores are not a closed-form function of $\cos(c_i,Q)$ but remain calibrated to query relevance, we argue that the same paradigm-level limitation holds in practice and confirm it empirically in Section~\ref{sec:experiments}, rather than treating it as a mere engineering artifact.

\subsection{Set Diversity-based Context Selection Methods}

To overcome the limitations of the chunk-wise paradigm, another line of research introduces set awareness. Classic approaches such as MMR \cite{carbonell1998mmr} and DPP \cite{kulesza2012dpp} achieve diverse selection through the maximal marginal relevance and determinant point processes, respectively. This set-level diversification objective has also been studied extensively in the database community, where result diversification is formalized as a combinatorial trade-off between relevance and coverage over query results \cite{drosou2010diversification, vieira2011diversification, qin2012diversify}; these works likewise treat ``diversity'' as an isotropic spread over the result space rather than coverage of an explicitly modeled demand distribution. SMART-RAG \cite{li2024smartrag} combines DPP with Natural Language Inference conflict modeling to further enhance the complementarity of information within sets. ScalDPP \cite{scaldpp2026} extends DPP to large-scale RAG scenarios by introducing P-Adapters and Diverse Margin Loss; however, it requires supervised training and thus falls outside the training-free setting we target. AdaGReS \cite{peng2025adagres} performs greedy set selection under token budget constraints and is closest to our work in terms of objective formulation. However, the direction of diversity in all the aforementioned methods is determined isotropically by the geometric structure of the embedding space, implicitly assuming that the query's information need radiates uniformly from a single point $\text{emb}(Q)$. In the terms of Section~\ref{sec:formulation}, these methods possess only an undirected form of Axis~B and lack Axis~A entirely: they know to spread selections apart, but not \emph{which} demand directions deserve coverage. Although AdaGReS also performs set-level selection, its demand remains a single query vector, making it fundamentally incapable of modeling the multidimensional sub-question structure of a query.

\subsection{Query Transformation and Expansion Methods}

Another direction aims to improve the recall quality by transforming or expanding queries. HyDE \cite{gao2022hyde}, a foundational work in query transformation, retrieves by embedding a hypothetical document generated by an LLM; as its sub-query is a single hypothetical document rather than a structured decomposition, its coverage of the query's information dimensions is only partial. Multi-Query Retrieval \cite{langchain2023multiquery} and RAG-Fusion \cite{raudaschl2023ragfusion} generate multiple sub-queries to broaden candidate recall; however, their final selection stage degrades into flat cosine ranking or Reciprocal Rank Fusion, whereby the multidimensional structural information embedded in the sub-queries is completely discarded during selection. OptiSet \cite{optiset2026} proposes an Expand-then-Refine framework that leverages changes in generator conditional utility to identify complementary and redundant evidence, but it relies on supervised training with LLMs, limiting its generalization to the training distribution. The fundamental distinction between GeoRAG and these works is that sub-queries in our method serve not as retrieval tools for expanding recall, but rather as modeling signals for constructing the multi-dimensional Information Demand Proxy distribution $\mathcal{P}_Q$. After reverse-validation quality weighting, they are encoded in $\mathcal{P}_Q$ to directly drive set-aware coverage selection, operating entirely without supervision or training.

\subsection{Optimal Transport Theory and Its Applications in NLP}

Optimal Transport has mature applications in NLP. Word Mover's Distance \cite{kusner2015wmd} applies it to document similarity measurement; WassRank \cite{yu2019wassrank} applies OT to listwise document ranking. The Sinkhorn algorithm \cite{cuturi2013sinkhorn} reduces the $O(n^3)$ complexity of exact Optimal Transport to a practical range via entropy regularization and has been widely applied in machine translation alignment and cross-modal matching. In this work, Optimal Transport is not the objective but the \emph{implementation}: the demand-weighted coverage objective of Section~\ref{sec:formulation} is realized through the entropic Sinkhorn--Wasserstein marginal coverage reduction $\Delta W_\varepsilon$, using the query's multi-dimensional Information Demand Proxy distribution $\mathcal{P}_Q$ as the transport target. As discussed in Section~\ref{sec:formulation}, as $\varepsilon\to 0$ the Sinkhorn objective becomes increasingly aligned with the hard coverage objective, so OT serves as a differentiable, numerically stable \emph{smoothed coverage proxy} for demand-driven set optimization---not an exact reformulation of it, and not a goal in itself.

Table~\ref{tab:method_comparison} summarizes the differences in capabilities between all methods.

\begin{table}[htbp]
\caption{Method Capability Comparison}
\label{tab:method_comparison}
\centering
\renewcommand{\arraystretch}{1.2}
\resizebox{\columnwidth}{!}{
\begin{tabular}{lccccc}
\toprule
\textbf{Method}
  & \textbf{Set-Aware}
  & \textbf{Multi-dim.}
  & \textbf{Sub-query}
  & \textbf{Prior Fusion}
  & \textbf{Train-Free} \\
\midrule
Cosine top-$k$     & $\times$     & $\times$     & $\times$     & $\times$     & $\checkmark$ \\
BGE-Reranker       & $\times$     & $\times$     & $\times$     & $\times$     & $\times$     \\
RankRAG$^\dagger$  & $\times$     & $\times$     & $\times$     & $\times$     & $\times$     \\
MMR                & $\checkmark$ & $\times$     & $\times$     & $\times$     & $\checkmark$ \\
DPP                & $\checkmark$ & $\times$     & $\times$     & $\times$     & $\checkmark$ \\
SMART-RAG          & $\checkmark$ & $\times$     & $\times$     & $\times$     & $\checkmark$ \\
ScalDPP$^\dagger$  & $\checkmark$ & $\times$     & $\times$     & $\times$     & $\times$     \\
AdaGReS            & $\checkmark$ & $\times$     & $\times$     & $\times$     & $\checkmark$ \\
HyDE               & $\times$     & $\times$     & $\triangle$  & $\times$     & $\checkmark$ \\
Multi-Query        & $\times$     & $\times$     & $\checkmark$ & $\times$     & $\checkmark$ \\
OptiSet$^\dagger$  & $\checkmark$ & $\triangle$  & $\checkmark$ & $\times$     & $\times$     \\
\midrule
\textbf{GeoRAG}    & $\checkmark$ & $\checkmark$ & $\checkmark$ & $\checkmark$ & $\checkmark$ \\
\bottomrule
\multicolumn{6}{l}{\small $^\dagger$ Requires training or fine-tuning. \quad $\triangle$ Partially satisfied.}
\end{tabular}
}
\end{table}

\section{Problem Formulation}
\label{sec:formulation}

\subsection{Implicit Optimization Objectives of Existing Methods and Their Limitations}

The implicit optimization objective of existing context selection methods can be uniformly formulated as follows:
\begin{equation}
    S^* = \underset{S \subseteq C,\, |S|=k}{\arg\max}\ \sum_{c \in S} f(c, Q)
\end{equation}
where $f(c_i, Q)$ is the chunk-wise relevance score between the candidate chunk $c_i$ and the query $Q$, such as cosine similarity or the cross-encoder score. This formulation is equivalent to assuming that the information need of a query can be fully expressed by a single embedding vector $\text{emb}(Q)$, that is, the need is a single point in the embedding space. Consequently, for a query requiring $T$ semantically independent sub-questions, $f(\cdot, Q)$ favors the dimension closest to $\text{emb}(Q)$, filling all $k$ slots with candidates for the same sub-question and leaving the remaining $T-1$ dimensions blank---directly manifested as the $0.82$ top-5 redundancy reported in Section~\ref{sec:introduction}.

\subsection{The Single-Point Coverage Barrier}
\label{sec:barrier}

We now show that the failure above is structural: it cannot be repaired with a better scoring model. We formalize the notion of a query-proximity-monotone selector and the bimodal demand of Figure~\ref{Motivation}, then prove non-coverability.

\begin{definition}[Query-proximity-monotone selector]
\label{def:singlepoint}
A selection method is a \emph{query-proximity-monotone selector} if it scores each candidate through a function of the candidate embedding and the single query embedding, $f(c_i,Q)=g\!\left(\text{emb}(c_i),\,\text{emb}(Q)\right)$, that is monotone nondecreasing in the cosine similarity $\cos(c_i,Q)$, and returns the top-$k$ candidates by this score. The cosine top-$k$ is the canonical instance. We emphasize that cross-encoder rerankers (e.g.\ BGE-Reranker, RankRAG) are \emph{not} formally members of this class, since their score is a joint encoding of the query--chunk pair rather than a function of $(\text{emb}(c_i),\text{emb}(Q))$ monotone in $\cos(c_i,Q)$; The Proposition~\ref{prop:noncover} does not bind them directly. We return to their practical behavior in Corollary~\ref{cor:capacity} and Section~\ref{sec:experiments}.
\end{definition}

\begin{assumption}[Bimodal demand with single-point asymmetry]
\label{ass:bimodal}
The candidate pool $C$ partitions into a peak-A subset $\mathcal{A}$ and a peak-B subset $\mathcal{B}$ (semantically independent dimensions of the query), each non-empty, with $|\mathcal{A}|\ge k$. Both peaks must be covered to answer $Q$. The query embedding is closer to peak~A, and the two peaks are separated in query-proximity:
\begin{equation}
    \min_{a\in\mathcal{A}}\cos(a,Q)\ >\ \max_{b\in\mathcal{B}}\cos(b,Q).
\end{equation}
\end{assumption}

This separation condition is precisely the geometry depicted in Figure~\ref{Motivation} and is the embedding-space counterpart of the empirically observed $0.82$ top-5 redundancy.

\begin{proposition}[Single-point non-coverability]
\label{prop:noncover}
Under Assumption~\ref{ass:bimodal}, every query-proximity-monotone selector (Definition~\ref{def:singlepoint}) returns a set $S$ with $S\subseteq\mathcal{A}$ and $S\cap\mathcal{B}=\varnothing$. That is, it leaves Peak~B completely uncovered, for any $k\le|\mathcal{A}|$ and for any choice of the score function $g$.
\end{proposition}

\begin{IEEEproof}
Because $g$ is monotone non-decreasing in $\cos(\cdot,Q)$, the score order it induces on candidates is consistent with the order by $\cos(\cdot,Q)$. By Assumption~\ref{ass:bimodal}, every $a\in\mathcal{A}$ has $\cos(a,Q)>\cos(b,Q)$ for every $b\in\mathcal{B}$, so every peak-A candidate is ranked strictly above every peak-B candidate. Since $|\mathcal{A}|\ge k$, the top-$k$ positions are filled exclusively by peak-A candidates. Hence $S\subseteq\mathcal{A}$ and $S\cap\mathcal{B}=\varnothing$.
\end{IEEEproof}

\begin{corollary}[Capacity-independence of the barrier]
\label{cor:capacity}
The conclusion of Proposition~\ref{prop:noncover} depends only on the monotonicity of $g$ in $\cos(\cdot,Q)$, not on its functional form. Increasing the capacity or training data of the scoring model changes the shape of $g$ but not this monotonicity, so no query-proximity-monotone selector---however powerful---can recover peak-B coverage. Only a selector that conditions on a multi-point demand representation and on the partial coverage of $S$ can escape the barrier.
\end{corollary}

Corollary~\ref{cor:capacity} is the structural upper bound referenced in Contribution~1: the limitation is a property of the single-point \emph{representation} on Axis~A, orthogonal to any engineering improvement of the scorer. Cross-encoder rerankers fall outside the formal hypothesis, yet they too score each chunk independently against the single query, lacking both a multi-point demand representation (Axis~A) and coverage-awareness of $S$ (Axis~B); we thus expect---and empirically confirm in Section~\ref{sec:experiments}---the same bimodal-coverage failure.

\subsection{Distributional Modeling of Information Needs}

To break the single-point assumption, we model the information need of a query as a discrete distribution defined over the candidate set.

\begin{definition}[Information Demand Proxy]
\label{def:proxy}
Given a query $Q$ and a candidate set $C = \{c_1, \ldots, c_K\}$, the Information Demand Proxy distribution $\mathcal{P}_Q$ of $Q$ is defined as:
\begin{equation}
    \mathcal{P}_Q : C \to [0,1], \quad \sum_{i=1}^{K} \mathcal{P}_Q(c_i) = 1
\end{equation}
where the weight $\mathcal{P}_Q(c_i)$ is jointly determined by two types of signals: global relevance $r_i$, which measures the general relevance of the retrieval between $c_i$ and $Q$, and local dimensional coverage $L_i$, which captures the contribution of $c_i$ to a specific sub-question dimension of $Q$. The construction of $\mathcal{P}_Q$ is detailed in Sections~\ref{sec:subquery} and \ref{sec:fusion}.
\end{definition}

\begin{definition}[Demand Coverage Set Optimization Objective]
\label{def:coverage}
Given a budget $k$, we reformulate context selection to maximize demand-weighted coverage of $\mathcal{P}_Q$ by the selected set:
\begin{equation}
\label{eq:facility}
    S^* = \underset{S \subseteq C,\, |S|=k}{\arg\max}\ F_Q(S),
    \quad
    F_Q(S) = \sum_{i=1}^{K} \mathcal{P}_Q(c_i)\,\max_{s \in S}\,\mathrm{sim}(c_i, s),
\end{equation}
where $\mathrm{sim}(c_i,s)=\max(\cos(c_i,s),0)$. Each unit of demand mass $\mathcal{P}_Q(c_i)$ is served by its best-matching selected chunk, so $F_Q$ is the demand-weighted facility-location coverage of $S$. Crucially, the coverage measure is anchored to the demand distribution $\mathcal{P}_Q$ and \emph{not} to the single query point: no $\cos(s,Q)$ term enters $F_Q$, so a chunk serving a far-from-$Q$ peak (Peak~B) contributes exactly as much as one serving a near peak, provided it reduces uncovered demand.
\end{definition}

\paragraph{Submodularity and the OT realization.}

\begin{lemma}
\label{lem:facility}
For a fixed demand point $c_i$, the function $g_i(S)=\max_{s\in S}\mathrm{sim}(c_i,s)$ with $\mathrm{sim}(c_i,s)=\max(\cos(c_i,s),0)\ge 0$ and $g_i(\varnothing)=0$ is non-decreasing monotone and submodular in $S$.
\end{lemma}

\begin{IEEEproof}
\emph{Monotonicity.} For $A\subseteq B$, the maximum over a larger set cannot decrease, so $g_i(A)\le g_i(B)$.
\emph{Submodularity.} Let $A\subseteq B\subseteq C$ and $e\in C\setminus B$. Write $a=\max_{s\in A}\mathrm{sim}(c_i,s)$ and $b=\max_{s\in B}\mathrm{sim}(c_i,s)$, so $a\le b$ by $A\subseteq B$. Then $g_i(A\cup\{e\})-g_i(A)=\max(0,\mathrm{sim}(c_i,e)-a)\ge\max(0,\mathrm{sim}(c_i,e)-b)=g_i(B\cup\{e\})-g_i(B)$, since $x\mapsto\max(0,x)$ is non-decreasing.
\end{IEEEproof}

\begin{proposition}
$F_Q(S)=\sum_{i=1}^{K}\mathcal{P}_Q(c_i)\,g_i(S)$ is non-decreasing monotone and submodular.
\end{proposition}
\begin{IEEEproof}
By Definition~\ref{def:proxy}, $\mathcal{P}_Q(c_i)\ge 0$. A non-negative weighted sum of monotone submodular functions is monotone submodular. By Lemma~\ref{lem:facility} each $g_i$ is monotone submodular, so $F_Q$ is monotone submodular. Consequently, greedy maximization achieves a $(1-1/e)$ approximation \cite{nemhauser1978}.
\end{IEEEproof}

Two consequences follow. \emph{First, non-redundancy is intrinsic}: by submodularity, the marginal gain of a chunk that duplicates demand already served by $S$ is automatically near zero---diminishing returns are a property of the objective, not an externally added penalty. \emph{Second}, the greedy algorithm of Section~\ref{sec:greedy} inherits the classical $(1-1/e)$ optimality guarantee \cite{nemhauser1978}.

We realize $F_Q$ through the entropic Sinkhorn--Wasserstein distance between $\mathcal{P}_Q$ and the demand-anchored uniform coverage measure of the set,
\begin{equation}
\label{eq:uniform}
    \mathcal{E}_S = \frac{1}{|S|}\sum_{s \in S} \delta_s,
\end{equation}
so that minimizing $W_\varepsilon(\mathcal{P}_Q,\mathcal{E}_S)$ acts as a \emph{smoothed coverage proxy} for maximizing $F_Q(S)$ rather than an exact reformulation: as $\varepsilon\to 0$ the transport plan concentrates each demand point onto its nearest selected chunk, so $-W_\varepsilon$ becomes increasingly aligned with $F_Q$ and shares its nearest-server assignment in the limit. The Sinkhorn form thus serves only as a differentiable, numerically stable surrogate for the marginal coverage gain, while the optimality guarantee is carried by the exact submodular objective $F_Q$.

The optimization objectives of the three types of methods are compared in Table~\ref{tab:objective_comparison}.

\begin{table}[htbp]
\caption{Comparison of Optimization Objectives Across Method Categories}
\label{tab:objective_comparison}
\centering
\renewcommand{\arraystretch}{1.3}
\resizebox{\columnwidth}{!}{
\begin{tabular}{lccc}
\toprule
\textbf{Method Category} & \textbf{Optimization Objective} & \textbf{Set-Aware} & \textbf{Demand Dimension} \\
\midrule
Chunk-wise Ranking
    & $\max (\sum f(c_i, Q))$
    & $\times$
    & Single Point \\
DPP and MMR
    & $\max (\det L) \ /\ \max (\text{MMR})$
    & $\checkmark$
    & Single Point, Isotropic \\
\textbf{GeoRAG}
    & $\max F_Q(S) \,\equiv\, \min W_\varepsilon(\mathcal{P}_Q, \mathcal{E}_S)$
    & $\checkmark$
    & \textbf{Multi-dimensional} \\
\bottomrule
\end{tabular}
}
\end{table}

\subsection{Three Core Challenges in Formal Optimization}
\label{sec:challenges}

Direct optimization of the objective in Definition~\ref{def:coverage} faces three independent core challenges, detailed as follows.

\textbf{P1: Unobservability of the Demand Distribution.} $\mathcal{P}_Q$ is a latent variable regarding the structural information need of the query. It cannot be directly inferred from a single vector $\text{emb}(Q)$, nor are annotated data available for supervised learning. The true sub-questions required by the query are completely unknown at inference time. We construct an approximate proxy through the generation of two-stage diverse sub-queries in Sections~\ref{sec:subquery} and \ref{sec:fusion}.

\textbf{P2: Unreliability of Sub-query Signals.} Even when using LLMs to generate sub-queries to approximate the multi-dimensional structure of $\mathcal{P}_Q$, approximately 10--20\% of the generated sub-queries exhibit semantic drift from the original query or high redundancy with pairwise cosine similarity exceeding $0.95$. Encoding these low-quality sub-queries into $\mathcal{P}_Q$ with equal weights will systematically distort the demand distribution and mislead the subsequent set-aware selection. We address this using a reverse-validation quality weighting mechanism in Section~\ref{sec:weighting}.

\textbf{P3: Intractability of Combinatorial Optimization.} Even if $\mathcal{P}_Q$ is known, the combinatorial space to precisely solve the optimal set is $\binom{200}{5} \approx 2.5 \times 10^{9}$. Furthermore, each evaluation of the entropic coverage surrogate requires calling an OT solver with a complexity of $O(n^3)$, making a brute-force search infeasible. In Section~\ref{sec:greedy}, the submodularity of $F_Q$ lets us reduce this problem to an iterative greedy selection process with linear complexity of $k$-steps and a $(1-1/e)$ guarantee.

\section{Method}
\label{sec:method}

GeoRAG jointly resolves the three core challenges formalized in Section~\ref{sec:challenges} through six stages, with the complete pipeline illustrated in Figure~\ref{fig:pipeline}.

\begin{figure*}[htbp]
    \centering
    \includegraphics[width=\linewidth]{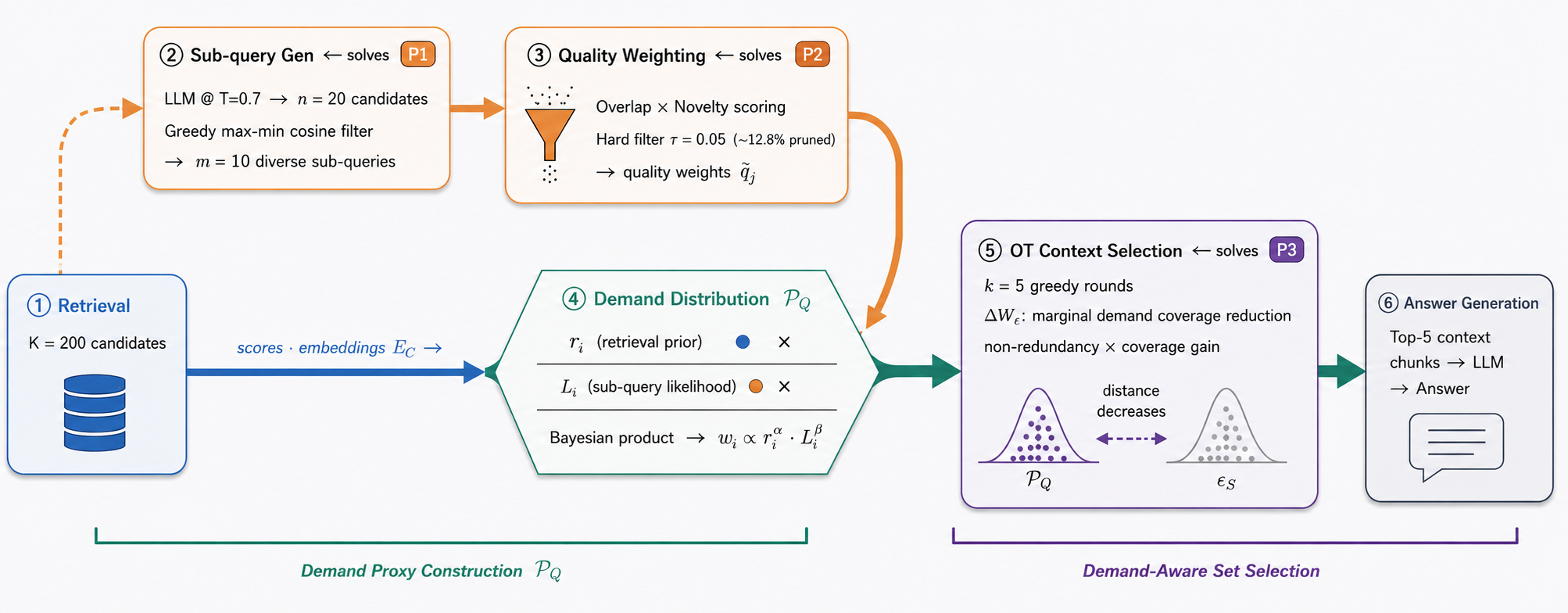}
    \caption{The six-stage pipeline of GeoRAG. }
    \label{fig:pipeline}
\end{figure*}

\subsection{Candidate Chunk Retrieval}

Given a query $Q$, we use Qwen3-Embedding-8B as the encoder to retrieve $K=200$ candidate chunks from a million-scale corpus via the FAISS IndexFlatIP index. This outputs the candidate set $C = \{c_1, \ldots, c_{200}\}$, the retrieval relevance scores $\text{score} \in \mathbb{R}^{200}$, and the embedding matrix $E_C \in \mathbb{R}^{200 \times d}$.

\subsection{Diverse Sub-query Generation and Quality Weighting}
\label{sec:subquery}

\subsubsection{Two-stage Diverse Sub-query Generation}

The single vector $\text{emb}(Q)$ is merely a single-point projection of the query's information need in the embedding space, failing to express its multi-dimensional structure. Multiple semantically dispersed sub-queries can approximate this structure, where each sub-query corresponds to a semantic sub-dimension of the query and collectively forms a multi-modal coverage. In the first stage, Qwen3-4B with a sampling temperature of $T=0.7$ generates $n=20$ candidate sub-queries for query $Q$ .
In the second stage, a greedy max-min cosine filter is applied to the 20 candidates to iteratively select $m=10$ sub-queries that are the most dispersed in the embedding space:
\begin{equation}
    q_{t+1} = \arg\max_{q \notin \mathcal{S}_t}\ \min_{s \in \mathcal{S}_t}\bigl(1 -
    \cos(q, s)\bigr).
\end{equation}
At each step, the candidate with the maximum cosine distance to the currently selected set $\mathcal{S}_t$ is chosen. This ensures that the 10 sub-queries uniformly cover multiple semantic dimensions of the query rather than clustering near the primary dimension. The UMAP visualization in Figure~\ref{fig:umap} confirms the improvement in semantic dispersion: the average pairwise cosine similarity drops from $0.87$ (single-stage) to $0.54$ (two-stage).

\begin{figure}[htbp]
    \centering
    \includegraphics[width=\linewidth]{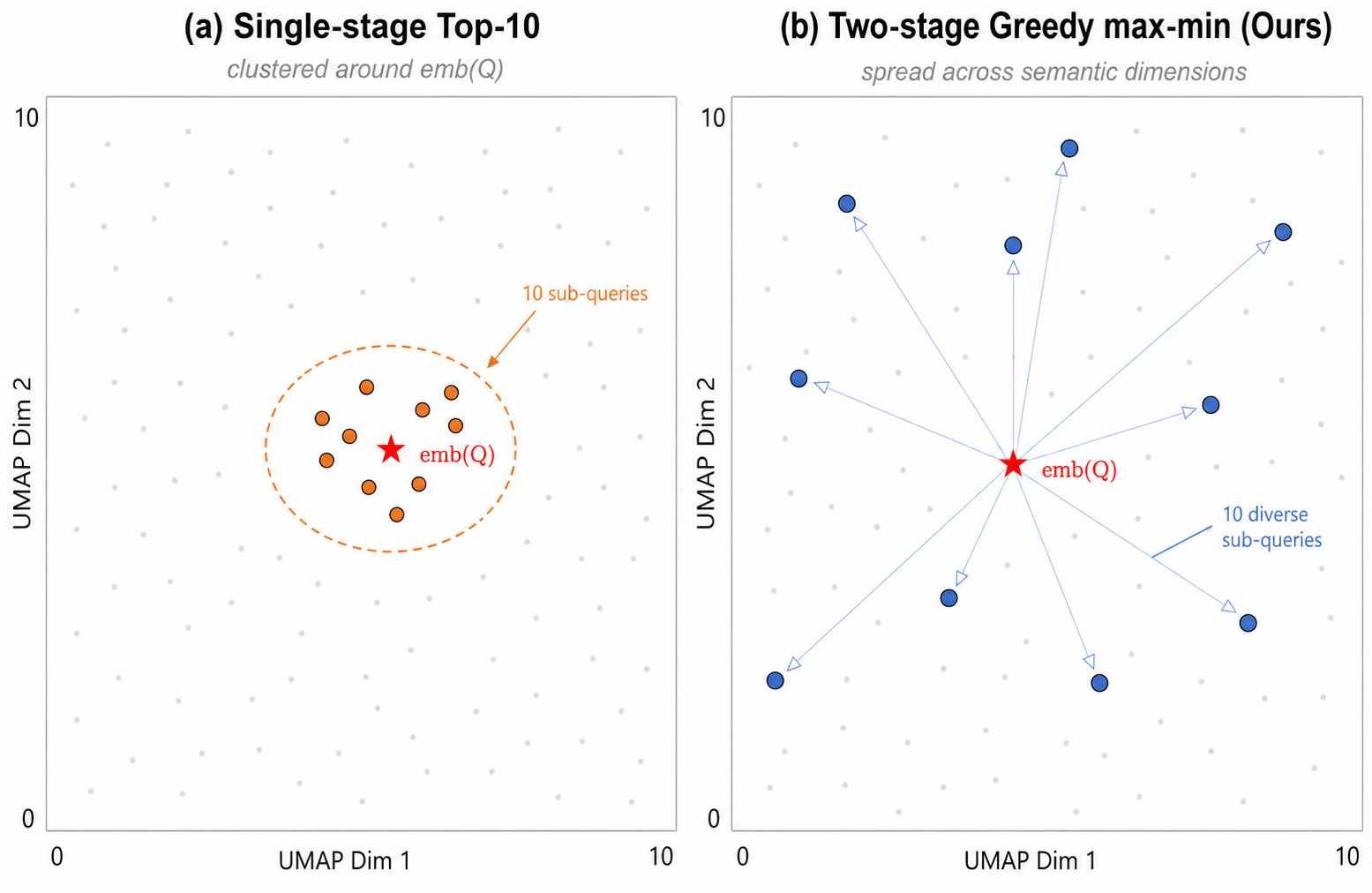}
    \caption{UMAP distribution comparison of sub-queries.}
    \label{fig:umap}
\end{figure}

\subsubsection{Reverse-Validation Quality Weighting}
\label{sec:weighting}

The quality of LLM-generated sub-queries varies significantly, with approximately 10--20\% exhibiting semantic drift from the original query or high redundancy with pairwise cosine similarity exceeding $0.95$. Encoding these directly into $\mathcal{P}_Q$ with uniform weights injects a systematic bias into the demand distribution. For each sub-query $q_j$, we independently retrieve the top-20 results $R_j$ using its embedding and compute two unsupervised quality metrics:
\begin{equation}
    \text{Quality}(q_j) = \underbrace{\frac{|R_j \cap C_{200}|}{20}}_{\text{Overlap}}
    \times \underbrace{\frac{|R_j \setminus \bigcup_{k \neq j} R_k|}{|R_j|}}_{\text{Novelty}}
\end{equation}
Overlap measures the recall alignment between the sub-query and the main candidate set, penalizing semantic drift since drifted sub-queries yield retrieval results with no intersection with $C_{200}$. Novelty measures the independent coverage increment brought by the sub-query, penalizing redundancy since sub-queries with highly overlapping retrieval results have zero independence. The product form requires the sub-query to satisfy both alignment and independence simultaneously; lacking either reduces the weight to zero.

We apply hard filtering to sub-queries with a quality score below the threshold $\tau_{\text{RV}}=0.05$ and normalize the remaining sub-queries into quality weighting coefficients:
\begin{equation}
    \tilde{q}_j = \frac{\text{Quality}(q_j)}{\sum_k \text{Quality}(q_k)}.
\end{equation}

\subsection{Construction of the Information Demand Proxy Distribution}
\label{sec:fusion}

Once the two-stage sub-queries $\{q_j\}_{j=1}^m$ and their quality weighting coefficients $\{\tilde{q}_j\}$ are available, we construct the Information Demand Proxy distribution $\mathcal{P}_Q$ by fusion of two complementary signals.

The general relevance between a candidate chunk $c_i$ and the query $Q$ is captured by the retrieval prior $r_i$, obtained by softmax normalization of the retrieval scores:
\begin{equation}
    r_i = \text{softmax}(\text{score}_i).
\end{equation}

The coverage contribution of the candidate chunk $c_i$ to each sub-dimension of the query is characterized by the quality-weighted sub-query likelihood $L_i$. Let $E_Q \in \mathbb{R}^{m \times d}$ denote the embedding matrix of the retained sub-queries $m$, and let $\tilde{\boldsymbol{q}} \in \mathbb{R}^{m}$ collect their quality weights. Then:
\begin{align}
    L_i &= \sum_{j=1}^{m} \tilde{q}_j \cdot \max\bigl(\cos(c_i, q_j),\ 0\bigr) \notag \\
        &\Longleftrightarrow\quad
         \boldsymbol{L} = \mathrm{clip}(E_C \cdot E_Q^\top,\ 0,\ 1)\cdot\tilde{\boldsymbol{q}}.
\end{align}

These two signals are fused via a Bayesian product formulation:
\begin{equation}
    w_i = \frac{r_i^{\alpha} \cdot L_i^{\beta}}{\sum_j r_j^{\alpha} \cdot L_j^{\beta}},
    \quad \alpha = \beta = 1.
\end{equation}
The product form requires candidate chunks to possess both global relevance and sub-dimensional coverage contribution simultaneously. Candidates that have high retrieval scores but cover no sub-questions, or are highly similar to a certain sub-query but have low overall retrieval relevance, are down-weighted, effectively filtering out noisy candidates. The design space of $\mathcal{P}_Q$ under different $(\alpha,\beta)$ settings is summarized in Table~\ref{tab:degradation}.

\begin{table}[htbp]
\caption{Degradation Spectrum of $\mathcal{P}_Q$ Under Different $\alpha$ and $\beta$ Settings}
\label{tab:degradation}
\centering
\renewcommand{\arraystretch}{1.3}
\resizebox{\columnwidth}{!}{
\begin{tabular}{cclc}
\toprule
$\alpha$ & $\beta$ & \textbf{Meaning} & \textbf{Ablation Variant} \\
\midrule
0 & 0 & Uniform Distribution          & No-information Lower Bound \\
1 & 0 & Pure Retrieval Prior           & Variant A, Main Ablation Baseline \\
0 & 1 & Pure Sub-query Likelihood      & Likelihood Only \\
\textbf{1} & \textbf{1} & Multi-dimensional Demand Proxy & \textbf{GeoRAG} \\
\bottomrule
\end{tabular}
}
\end{table}

\subsection{Greedy Coverage Set Selection}
\label{sec:greedy}

After $\mathcal{P}_Q$ is constructed, we solve the demand coverage objective of Definition~\ref{def:coverage}. Since exact optimization is intractable due to P3, we exploit the submodularity of $F_Q$ and employ greedy maximization: at each step, we add the chunk with the maximum marginal coverage gain $\Delta F_Q(c_i\mid S)=F_Q(S\cup\{c_i\})-F_Q(S)$ to the selected set $S$, repeating this $k=5$ times. The full procedure is given in Algorithm~\ref{alg:greedy}. In the entropic realization, the marginal gain is estimated by the Sinkhorn--Wasserstein marginal coverage reduction in the following factorized surrogate form:
\begin{equation}
    c^* = \arg\max_{c_i \notin S}\
    \underbrace{\bigl(1 - \max_{s \in S}\cos(c_i, s)\bigr)}_{\text{reassignment gate}}
    \times
    \underbrace{\Delta W_\varepsilon(c_i \mid S,\ \mathcal{P}_Q)}_{\text{reclaimed demand mass}}
\end{equation}

The reduction in demand coverage $\Delta W_\varepsilon$ measures the extent to which the addition $c_i$ brings the coverage distribution $\mathcal{E}_S$ closer to $\mathcal{P}_Q$:
\begin{equation}
    \Delta W_\varepsilon(c_i \mid S, \mathcal{P}_Q) = W_\varepsilon(\mathcal{P}_Q,\
    \mathcal{E}_S) - W_\varepsilon(\mathcal{P}_Q,\ \mathcal{E}_{S \cup \{c_i\}}).
\end{equation}

\begin{algorithm}[t]
\caption{Greedy Demand-Coverage Selection}
\label{alg:greedy}
\begin{algorithmic}[1]
\STATE \textbf{Input:} candidates $C$, demand proxy $\mathcal{P}_Q$, budget $k$, regularizer $\varepsilon$, query $Q$
\STATE \textbf{Output:} selected set $S$ with $|S|=k$
\STATE $s_1 \leftarrow \arg\max_{c_i\in C}\cos(c_i,Q)$ \quad // cold start
\STATE $S \leftarrow \{s_1\}$
\FOR{$t = 2$ \TO $k$}
    \FORALL{$c_i \in C\setminus S$}
        \STATE $\text{gate}_i \leftarrow 1 - \max_{s\in S}\cos(c_i,s)$
        \STATE $\Delta W_i \leftarrow W_\varepsilon(\mathcal{P}_Q,\mathcal{E}_S) - W_\varepsilon(\mathcal{P}_Q,\mathcal{E}_{S\cup\{c_i\}})$
        \STATE $\text{score}_i \leftarrow \text{gate}_i \cdot \Delta W_i$
    \ENDFOR
    \STATE $c^* \leftarrow \arg\max_{c_i\in C\setminus S}\text{score}_i$
    \STATE $S \leftarrow S \cup \{c^*\}$
\ENDFOR
\RETURN $S$
\end{algorithmic}
\end{algorithm}

In the Sinkhorn implementation, the cost matrix uses the cosine distance $M_{ij} = 1 - \cos(c_i, s_j)$, with regularization coefficient $\varepsilon = 0.01$. A log-domain implementation is adopted for numerical stability (Algorithm~\ref{alg:sinkhorn}), with a maximum of 100 iterations. When $S = \emptyset$ on cold start, the first chunk is initialized using $\arg\max_i \cos(c_i, Q)$ to prevent coverage estimation from degenerating due to an empty selected set.

\begin{algorithm}[t]
\caption{Log-domain Sinkhorn for $W_\varepsilon(\mathcal{P}_Q,\mathcal{E}_S)$}
\label{alg:sinkhorn}
\begin{algorithmic}[1]
\STATE \textbf{Input:} cost $M\in\mathbb{R}^{K\times|S|}$, marginals $\mathbf{p}=\mathcal{P}_Q$, $\mathbf{e}=\mathcal{E}_S$, $\varepsilon$, iters$=100$
\STATE $\mathbf{f}\leftarrow\mathbf{0}_K,\ \mathbf{g}\leftarrow\mathbf{0}_{|S|}$
\FOR{$\ell=1$ \TO iters}
    \STATE $f_i \leftarrow \varepsilon\log p_i - \varepsilon\log\!\sum_j \exp\!\bigl(\tfrac{g_j-M_{ij}}{\varepsilon}\bigr)$
    \STATE $g_j \leftarrow \varepsilon\log e_j - \varepsilon\log\!\sum_i \exp\!\bigl(\tfrac{f_i-M_{ij}}{\varepsilon}\bigr)$
\ENDFOR
\STATE $\pi_{ij}\leftarrow \exp\!\bigl(\tfrac{f_i+g_j-M_{ij}}{\varepsilon}\bigr)$
\RETURN $W_\varepsilon=\sum_{ij}\pi_{ij}M_{ij}$
\end{algorithmic}
\end{algorithm}

\paragraph{Why the two factors are both intrinsic.}
The two factors above are not an objective plus an external regularizer; they are the two facets of the same submodular marginal gain $\Delta F_Q(c_i\mid S)$. The marginal gain of the facility-location objective~\eqref{eq:facility} is non-zero only where $c_i$ becomes the new nearest server for some demand region and equals the demand mass it reclaims there. The first factor is the \emph{reassignment gate}---it is small exactly when $c_i$ duplicates an already-selected chunk, which is the facility-location condition that intrinsically suppresses redundancy---and the second factor, computed via Sinkhorn, is the \emph{reclaimed demand mass}. Their product is therefore a tractable estimate of $\Delta F_Q(c_i\mid S)$ in which redundancy suppression is a property of the coverage objective itself, consistent with the submodularity established in Section~\ref{sec:formulation}.

The core of the set-aware mechanism (Axis~B) is that once a certain semantic dimension in $\mathcal{P}_Q$ is sufficiently covered by $\mathcal{E}_S$, the marginal gain of any new chunk in that dimension approaches zero, and the greedy selection automatically shifts toward the direction with the largest coverage gap. Because $F_Q$ is monotone submodular, greedy maximization of $F_Q$ enjoys the classical $(1-1/e)$ optimality guarantee \cite{nemhauser1978}; the factorized Sinkhorn surrogate we run provides a smoothed estimate of each marginal gain $\Delta F_Q(c_i\mid S)$, with the entropic regularizer $\varepsilon$ controlling how closely the proxy tracks the exact gain (Section~\ref{sec:formulation}).

\section{Experiments}
\label{sec:experiments}

\subsection{Experimental Settings}

\textbf{Datasets.} The experiments cover six open-domain QA benchmarks, grouped by task type, as summarized in Table~\ref{tab:datasets}.

\begin{table}[htbp]
\caption{Overview of Evaluation Benchmarks}
\label{tab:datasets}
\centering
\renewcommand{\arraystretch}{1.3}
\resizebox{\columnwidth}{!}{
\begin{tabular}{llcc}
\toprule
\textbf{Type} & \textbf{Datasets} & \textbf{Test Set Size} & \textbf{Metrics} \\
\midrule
Single-hop Factual   & NQ, TriviaQA        & 3,610 / 11,313 & EM, F1   \\
Multi-hop Reasoning  & HotpotQA, 2WikiMHQA & 7,405 / 12,576 & EM, F1   \\
Ambiguous Open-ended & ASQA                & 948            & EM, F1   \\
Fact Verification    & FEVER               & 9,999          & Accuracy \\
\bottomrule
\end{tabular}
}
\end{table}

\textbf{Models and Corpus.} Qwen3-Embedding-8B is used as the embedding model with FAISS IndexFlatIP. Both the generator and the sub-query generation utilize Qwen3-4B. The corpus is a 1-million-chunk subset of the December 2018 Wikipedia dump, obtained by retaining, for every query in our benchmarks, the union of all gold supporting passages together with their top-ranked dense neighbors, then deduplicating to 1M 100-word chunks. The gold passages are used only to construct this compact evaluation corpus, never by any retriever or selection method at inference time: all methods receive only the query and operate on the identical corpus. To confirm that GeoRAG's selection gain does not depend on this compact-corpus construction, Section~\ref{sec:fullwiki} additionally reports a full-Wikipedia (no-gold-injection) experiment. On ASQA we report the standard STR-EM (string Exact Match against the set of acceptable short answers); for brevity we denote it ``EM'' in all tables. GraphRAG \cite{edge2024graphrag} is used here purely as a retriever: we take the passage-level nodes it surfaces for a query (ranked by its community-and-entity relevance score) as the top-$200$ candidate pool, discarding its summarization stage, so that all selection methods feed an identical $200$-chunk pool. The Recall@200 of each retriever over the six benchmarks is reported in Table~\ref{tab:recall200}.

\begin{table}[htbp]
\caption{Recall@200 of Six Retrieval Methods Averaged over Six Benchmarks}
\label{tab:recall200}
\centering
\renewcommand{\arraystretch}{1.3}
\resizebox{\columnwidth}{!}{
\begin{tabular}{lc}
\toprule
\textbf{Retrieval Method} & \textbf{Recall@200 (avg.\ over 6 benchmarks)} \\
\midrule
Dense (Qwen3-Embedding-8B) & 93.8 \\
BM25                       & 88.7 \\
Hybrid RRF                 & 92.1 \\
HyDE                       & 91.4 \\
MultiQuery                 & 90.6 \\
GraphRAG                   & 89.3 \\
\bottomrule
\end{tabular}
}
\end{table}

\textbf{Hyperparameters.} The default configuration uses $K=200$, $k=5$, $T=0.7$, $n=20$, $m=10$, $\alpha=\beta=1$, $\varepsilon=0.01$, and $\tau_{\text{RV}}=0.05$. The Pareto-optimal choice of $K=200$ is validated in Section~\ref{sec:k_robustness}.

\textbf{Statistical Testing.} We performed 1,000 bootstrap resamples for each dataset and reported 95\% confidence intervals. The results on NQ, HotpotQA, 2WikiMHQA and ASQA satisfy $p < 0.05$, while TriviaQA and FEVER yield $p > 0.1$, denoted by $\dagger$ in the results tables.

\subsection{Main Results Across Retrieval Methods}

We construct a complete experimental matrix of 6 retrieval methods, pairing direct top-5 selection against top-200 followed by GeoRAG selection to top-5, validating GeoRAG as a retrieval-agnostic post-processing module (Table~\ref{tab:main_results}).

\begin{table*}[htbp]
\caption{Main Experimental Results. Scores are EM across five benchmarks and accuracy for
FEVER, reported over six retrieval methods. $\dagger$ denotes $p > 0.1$ for TriviaQA and
FEVER columns. F1 results follow the same trend across all cells and are omitted for space.}
\label{tab:main_results}
\centering
\renewcommand{\arraystretch}{1.3}
\resizebox{\textwidth}{!}{
\begin{tabular}{llccccccc}
\toprule
\textbf{Method} & \textbf{Candidates} & \textbf{NQ} & \textbf{TriviaQA$^\dagger$}
& \textbf{HotpotQA} & \textbf{2WikiMHQA} & \textbf{ASQA} & \textbf{FEVER$^\dagger$}
& \textbf{Avg} \\
\midrule
Dense (direct)      & top-5             & 44.1  & 65.2 & 33.8 & 29.4 & 36.7 & 74.1 & 47.2 \\
Dense + GeoRAG      & top-200$\to$5     & 54.8  & 70.4 & 43.5 & 37.8 & 46.1 & 75.3 & 54.7 \\
$\Delta$            &                   & \textbf{+10.7} & +5.2 & \textbf{+9.7} & +8.4 & +9.4 & +1.2 & \textbf{+7.5} \\
\midrule
BM25 (direct)       & top-5             & 42.8  & 66.8 & 31.2 & 27.3 & 34.2 & 73.9 & 46.0 \\
BM25 + GeoRAG       & top-200$\to$5     & 52.6  & 71.2 & 40.8 & 35.9 & 43.9 & 74.8 & 53.2 \\
$\Delta$            &                   & +9.8  & +4.4 & +9.6 & +8.6 & \textbf{+9.7} & +0.9 & +7.2 \\
\midrule
Hybrid RRF (direct) & top-5             & 46.2  & 67.4 & 35.4 & 31.1 & 38.9 & 76.1 & 49.2 \\
Hybrid RRF + GeoRAG & top-200$\to$5     & 55.2  & 71.8 & 44.2 & 39.3 & 47.3 & 76.6 & 55.7 \\
$\Delta$            &                   & +9.0  & +4.4 & +8.8 & +8.2 & +8.4 & +0.5 & +6.5 \\
\midrule
HyDE (direct)       & top-5             & 47.1  & 65.8 & 36.1 & 30.8 & 39.8 & 73.2 & 48.8 \\
HyDE + GeoRAG       & top-200$\to$5     & 56.1  & 70.1 & 44.9 & 39.1 & 48.1 & 74.1 & 55.4 \\
$\Delta$            &                   & +9.0  & +4.3 & +8.8 & +8.3 & +8.3 & +0.9 & +6.6 \\
\midrule
MultiQuery (direct) & top-5             & 45.9  & 66.3 & 37.2 & 32.1 & 40.6 & 74.6 & 49.5 \\
MultiQuery + GeoRAG & top-200$\to$5     & 54.8  & 70.9 & 46.1 & 41.2 & 49.4 & 75.8 & 56.4 \\
$\Delta$            &                   & +8.9  & +4.6 & +8.9 & \textbf{+9.1} & +8.8 & +1.2 & +6.9 \\
\midrule
GraphRAG (direct)   & top-5             & 46.3  & 65.9 & 37.8 & 33.4 & 38.1 & 76.8 & 49.7 \\
GraphRAG + GeoRAG   & top-200$\to$5     & 55.4  & 70.6 & 45.7 & 42.1 & 47.6 & 77.2 & 56.4 \\
$\Delta$            &                   & +9.1  & +4.7 & +7.9 & +8.7 & +9.5 & +0.4 & +6.7 \\
\bottomrule
\end{tabular}
}
\end{table*}

GeoRAG yields consistent $+6.5$ to $+7.5$ EM gains across all six retrieval methods, with paired bootstrap $p<0.05$ on every benchmark except TriviaQA and FEVER. The gain correlates positively with candidate set homogeneity: Dense ($+7.5$) and BM25 ($+7.2$) have the most redundant candidate sets and benefit most from GeoRAG's multi-dimensional coverage optimization, whereas Hybrid RRF ($+6.5$) already fuses lexical and semantic signals and starts from higher diversity. This confirms our mechanistic hypothesis: the benefit comes from set selection within the candidate pool rather than from better initial retrieval. As reported in Table~\ref{tab:recall200}, Recall@200 ranges from 88.7\% to 93.8\% with a Spearman correlation of $\rho = -0.23$ ($n=6$, not significant) against the gain, indicating that the evidence is already in the pool and GeoRAG's role is to stop discarding it at selection time.

The method-specific strengths are amplified under GeoRAG. HyDE+GeoRAG is best in NQ (56.1 EM); MultiQuery+GeoRAG is best on HotpotQA (46.1) and ASQA (49.4), as its sub-queries cover candidates for different reasoning steps; GraphRAG+GeoRAG is best in 2WikiMHQA (42.1) and FEVER (77.2), since graph-structure awareness aligns with Wikipedia entity chains. The small FEVER gain ($+0.4$ to $+1.2$) marks the method's applicability boundary: as a single-proposition truth-value task, its information need is concentrated on one dimension. The universal lift over truncation, together with the absence of a positive gain--Recall@200 correlation, confirms that GeoRAG acts as a retrieval-agnostic \emph{selection} module, with gains largest exactly where the pool is most redundant---the regime the single-point barrier predicts.

\subsection{Comparison with Selection Baselines}
\label{sec:baselines}

The previous experiment isolates GeoRAG's gain over naive truncation. To verify that the gain survives comparison with the selection methods criticized in Section~\ref{sec:formulation}, we fix the retriever to Dense and the candidate pool to the same top-$200$, and replace only the selection-of-$5$ strategy. We compare against two isotropic set-diversity methods (MMR~\cite{carbonell1998mmr}, DPP~\cite{kulesza2012dpp}), a strong cross-encoder reranker (BGE-Reranker~\cite{xiao2023bge}), a conflict-sensitive variant of DPP (SMART-RAG~\cite{li2024smartrag}), and the closest greedy set-selection competitor (AdaGReS~\cite{peng2025adagres}). The results are reported in Table~\ref{tab:baselines}.

GeoRAG achieves the best EM on every benchmark (tying BGE-Reranker on FEVER) and an average-EM margin of $+1.9$ over the strongest baseline (AdaGReS). Two observations are central to our claims. First, BGE-Reranker reaches $52.8$ NQ EM---GeoRAG exceeds it by $+2.0$ EM while incurring $35\%$ lower selection latency ($1{,}222$ vs.\ $1{,}890$ ms), and the gap \emph{widens} on multi-hop HotpotQA ($+3.9$) and 2WikiMHQA ($+3.9$), exactly the regime where the single-point barrier of Proposition~\ref{prop:noncover} predicts cross-encoder failure. This is the empirical corroboration of the analysis in Section~\ref{sec:barrier}. Second, the isotropic diversity methods (MMR, DPP) improve over truncation but plateau well below GeoRAG, because they possess only an undirected form of Axis~B and no Axis~A: spreading selections apart in embedding space cannot recover a demand direction the query never points to. AdaGReS, the only baseline with genuine greedy set selection, is the closest competitor, yet it still lacks the multi-dimensional demand model ($\mathcal{P}_Q$), which accounts for its multi-hop shortfall. The advantage is smallest on the uni-dimensional FEVER.
\begin{table*}[htbp]
\caption{Comparison with selection baselines under the Dense retriever (top-$200\to5$,
identical candidate pool and generator).}
\label{tab:baselines}
\centering
\renewcommand{\arraystretch}{1.3}
\resizebox{\textwidth}{!}{
\begin{tabular}{lcccccccr}
\toprule
\textbf{Selection Method} & \textbf{NQ} & \textbf{TriviaQA$^\dagger$}
& \textbf{HotpotQA} & \textbf{2WikiMHQA} & \textbf{ASQA} & \textbf{FEVER$^\dagger$}
& \textbf{Avg} & \textbf{Sel.\ Lat.\ (ms)} \\
\midrule
Cosine top-5 (truncation)        & 44.1 & 65.2 & 33.8 & 29.4 & 36.7 & 74.1 & 47.2 & \textbf{0.4} \\
MMR \cite{carbonell1998mmr}      & 47.5 & 66.8 & 37.2 & 32.0 & 39.8 & 74.3 & 49.6 & 14.6 \\
DPP \cite{kulesza2012dpp}        & 48.3 & 67.1 & 38.0 & 32.8 & 40.6 & 74.4 & 50.2 & 41.8 \\
SMART-RAG \cite{li2024smartrag}  & 50.1 & 68.2 & 40.3 & 34.6 & 42.1 & 74.6 & 51.7 & 118.5 \\
BGE-Reranker \cite{xiao2023bge}  & 52.8 & 69.5 & 39.6 & 33.9 & 41.2 & \textbf{75.3} & 52.1 & 1890.0 \\
AdaGReS \cite{peng2025adagres}   & 51.9 & 69.1 & 41.5 & 35.8 & 43.4 & 74.9 & 52.8 & 96.3 \\
\textbf{GeoRAG}                  & \textbf{54.8} & \textbf{70.4} & \textbf{43.5}
& \textbf{37.8} & \textbf{46.1} & \textbf{75.3} & \textbf{54.7} & 1222.0 \\
\bottomrule
\multicolumn{9}{l}{\small $^\dagger$ $p>0.1$ for GeoRAG vs.\ AdaGReS on TriviaQA and FEVER; $p<0.05$ on all other columns.}
\end{tabular}
}
\end{table*}

\subsection{Robustness to the Corpus: Full-Wikipedia (No-Gold-Injection) Test}
\label{sec:fullwiki}

The main experiments use a compact 1M-chunk corpus that, by construction, retains every query's gold supporting passages, which guarantees high Recall@200. A natural concern is whether GeoRAG's reported gain is partly an artifact of this ``recall-easy'' construction: if the evidence were not guaranteed to be in the pool, the bottleneck might revert from \emph{selection} to \emph{recall}. To rule this out, we re-run the full pipeline on the \emph{entire} FlashRAG~\cite{flashrag} \texttt{wiki18\_100w} Wikipedia collection ($\sim$21M passages) with \emph{no} gold injection: the candidate pool is produced solely by the Dense retriever over the full index, so gold evidence enters the top-$200$ only if the retriever actually surfaces it. We report HotpotQA and 2WikiMHQA, the two multi-hop benchmarks where the single-point barrier predicts the largest selection effect, in Table~\ref{tab:fullwiki}.

\begin{table}[htbp]
\caption{Full-Wikipedia, no-gold-injection test.}
\label{tab:fullwiki}
\centering
\renewcommand{\arraystretch}{1.3}
\resizebox{\columnwidth}{!}{
\begin{tabular}{lcccc}
\toprule
\textbf{Setting} & \textbf{Recall@200} & \textbf{Cosine top-5}
& \textbf{GeoRAG} & $\Delta$ \\
\midrule
\multicolumn{5}{l}{\emph{HotpotQA}} \\
\quad Compact 1M (gold-injected) & 91.7 & 33.8 & 43.5 & $+9.7$ \\
\quad Full Wiki (no gold)        & 71.4 & 27.9 & 35.6 & $+7.7$ \\
\midrule
\multicolumn{5}{l}{\emph{2WikiMHQA}} \\
\quad Compact 1M (gold-injected) & 90.3 & 29.4 & 37.8 & $+8.4$ \\
\quad Full Wiki (no gold)        & 68.2 & 23.6 & 30.5 & $+6.9$ \\
\bottomrule
\end{tabular}
}
\end{table}

As expected, removing gold injection lowers Recall@200 substantially ($91.7\!\to\!71.4$ on HotpotQA, $90.3\!\to\!68.2$ on 2WikiMHQA) and drags down the absolute EM of \emph{both} selectors, confirming that recall does matter in the harder setting. Crucially, however, the gain in GeoRAG \emph{selection} over Cosine top-5 persists ($+7.7$ on HotpotQA, $+6.9$ on 2WikiMHQA), retaining the large majority of the gain observed on the compact corpus. The gain shrinks only modestly, consistent with the mechanism: when the retriever surfaces both demand peaks, GeoRAG still stops the single-point selector from spending all slots on Peak~A; the residual loss comes purely from queries where the harder retrieval setting fails to recall Peak~B at all---a recall failure outside the scope of any selection method. This decouples our central claim from the corpus construction: the selection bottleneck GeoRAG addresses is real even when the evidence is \emph{not} guaranteed to be present.

\subsection{Necessity of Multi-dimensional Demand Distribution Modeling}

In this section, we systematically vary the construction of $\mathcal{P}_Q$ to quantify the independent contribution of multi-dimensional demand modeling compared to single-point priors, as reported in Figure~\ref{fig:demand_ablation}.

\begin{figure}[htbp]
    \centering
    \includegraphics[width=\linewidth]{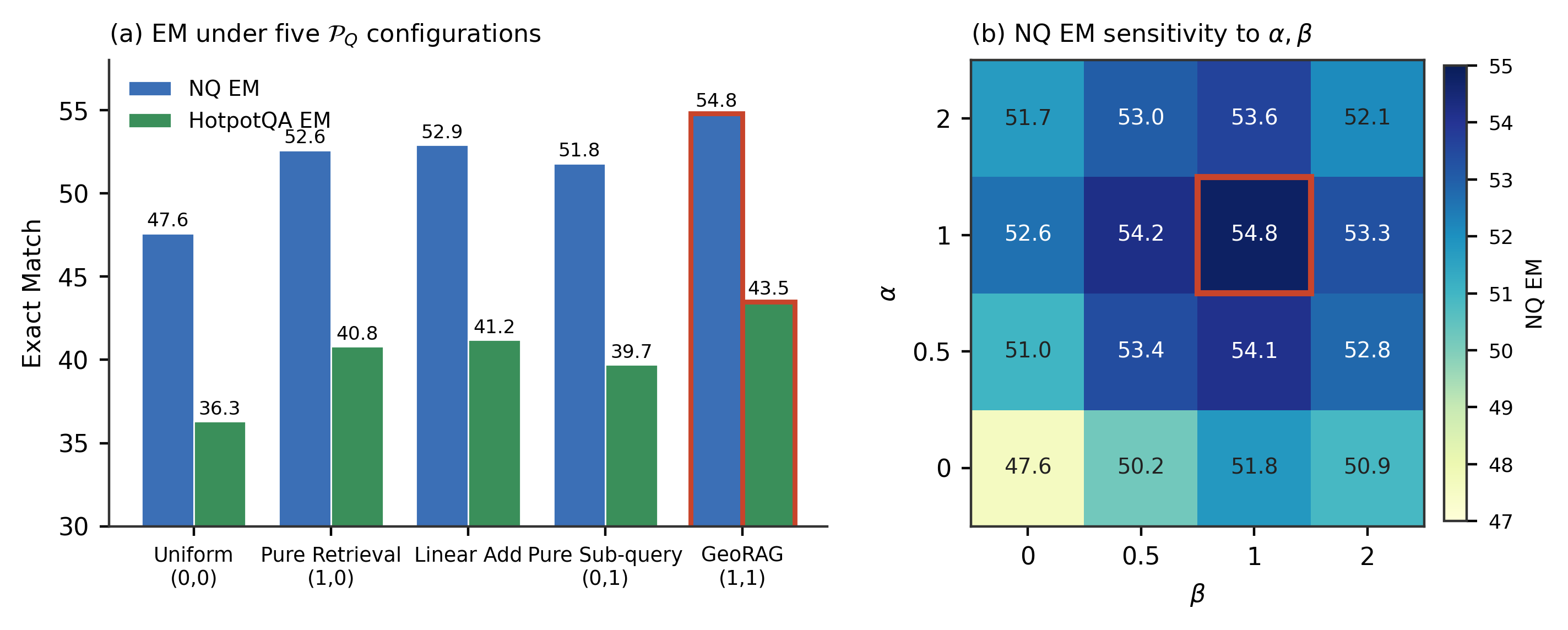}
    \caption{Demand distribution ablation.}
    \label{fig:demand_ablation}
\end{figure}

The degradation from the uniform distribution results in a drop of $7.2$ EM: non-informative weights prevent the coverage estimate from distinguishing the semantic priorities of the candidate chunks, causing the transport plan to degenerate into random coverage. This shows that $\mathcal{P}_Q$ is the key driver of the entire framework (Axis~A). Under the same set-aware framework, replacing $\mathcal{P}_Q$ from a single-point retrieval prior ($\beta=0$) to a Bayesian posterior fused with sub-query likelihood ($\beta=1$) improves HotpotQA by $+2.7$ EM: sub-query decomposition effectively captures the weights of sub-question directions in multi-hop queries. The product form contributes an additional $+1.9$ EM on NQ ($+2.3$ EM on HotpotQA) over linear addition, because additive fusion allows chunks strong in only a single signal to pass through, whereas the multiplicative form effectively filters out such interfering candidates. As shown in the heat map of Figure~\ref{fig:demand_ablation}(b), the NQ EM remains at or above 53.4 within the range $\alpha = \beta \in [0.5, 1.0]$, indicating insensitivity to selection of hyperparameters.

\subsection{Irreplaceability of Set-Aware Coverage Selection}

This section strictly isolates the independent contribution of set-aware coverage selection (Axis~B). The three compared methods share completely identical sub-query generation, reverse validation, and construction of $\mathcal{P}_Q$. The sole difference lies in the strategy to select $k=5$ chunks from the top-200 candidates, as reported in Table~\ref{tab:ot_ablation}.

\begin{table}[htbp]
\caption{Set-Aware Coverage Contribution Isolation.}
\label{tab:ot_ablation}
\centering
\renewcommand{\arraystretch}{1.3}
\resizebox{\columnwidth}{!}{
\begin{tabular}{lccccr}
\toprule
\textbf{Method} & \textbf{NQ EM} & \textbf{HotpotQA EM} & \textbf{2WikiMHQA EM}
& $\Delta$ \textbf{vs GeoRAG} & \textbf{Extra Lat. (ms)} \\
\midrule
SubQ-MaxCos              & 51.5 & 39.2 & 33.9 & $-3.3$ & $+1161$ \\
SubQ-AvgCos              & 52.2 & 40.1 & 34.7 & $-2.6$ & $+1162$ \\
\textbf{GeoRAG}          & \textbf{54.8} & \textbf{43.5} & \textbf{37.8}
& \textbf{---} & $\mathbf{+1222}$ \\
Set-aware coverage gain  & $+2.6$ & $+3.4$ & $+3.1$ & --- & $+60$ \\
\bottomrule
\end{tabular}
}
\end{table}

The scores $L_i$ of SubQ-AvgCos are completely identical to the $\mathcal{P}_Q$ weights in GeoRAG. The difference is that SubQ-AvgCos independently scores and ranks each chunk, whereas GeoRAG maximizes demand coverage in a set-aware manner via the Sinkhorn realization of the facility-location objective. Even when sharing the same demand signal, removing set-aware coverage still results in a $2.6$ EM loss on NQ, and the gap further widens in multi-hop scenarios ($+3.4$ on HotpotQA, $+3.1$ on 2WikiMHQA). The fundamental reason is that independent ranking cannot perceive the coverage status of the selected set and tends to repeatedly select high-scoring chunks in the same semantic direction. The 60 ms extra latency in exchange for a consistent $+2.6$ EM gain proves that set-aware coverage selection cannot be replaced by any chunk-wise scoring scheme.

\subsection{Exact Facility-Location Greedy vs.\ the Sinkhorn Surrogate}
\label{sec:exact_vs_sinkhorn}

Section~\ref{sec:formulation} proves that the coverage objective $F_Q$ (Eq.~\eqref{eq:facility}) is monotone submodular, so its \emph{exact} margin gain $\Delta F_Q(c_i\mid S)$ has a closed form computable in $O(Kk)$ without any optimal-transport machinery, and greedy maximization of this exact gain is the variant to which the classical $(1-1/e)$ guarantee provably applies. The deployed GeoRAG selector instead ranks the candidates by the factorized Sinkhorn surrogate $(\text{gate})\times\Delta W_\varepsilon$ of Section~\ref{sec:greedy} plus a cold start of cosine, which is an \emph{approximation} of $\Delta F_Q$ rather than $\Delta F_Q$ itself. To quantify the cost or benefit of this approximation and to test whether the optimal-transport layer is load-bearing or merely a differentiable convenience, we replace \emph{only} the per-step selection rule while holding the Dense top-$200$ pool, the demand proxy $\mathcal{P}_Q$, the budget $k=5$, and the generator fixed. We additionally report the realized objective value $F_Q(S)$ of the returned set, so the rules are compared on the exact quantity the theory concerns rather than on downstream EM alone, and sweep $\varepsilon$ to make the $\varepsilon\to0$ limit (where the surrogate should approach exact greedy) visible. The results are in Table~\ref{tab:exact_vs_sinkhorn} and Table~\ref{tab:eps_sweep}.

\begin{table}[htbp]
\caption{Exact facility-location greedy vs.\ the Sinkhorn surrogate . }
\label{tab:exact_vs_sinkhorn}
\centering
\renewcommand{\arraystretch}{1.3}
\resizebox{\columnwidth}{!}{
\begin{tabular}{lcccccr}
\toprule
\textbf{Selection rule} & \textbf{NQ} & \textbf{HotpotQA} & \textbf{2WikiMHQA}
& \textbf{Avg EM} & \textbf{$F_Q(S)$} & \textbf{Sel.-rule Lat.\ (ms)} \\
\midrule
Exact-FL greedy            & 54.4 & 42.8 & 37.0 & 44.7 & \textbf{0.842} & \textbf{0.6} \\
$\Delta W_\varepsilon$ only (no gate) & 48.3 & 37.0 & 31.8 & 39.0 & 0.781 & 49.1 \\
\textbf{Sinkhorn (GeoRAG)} & \textbf{54.8} & \textbf{43.5} & \textbf{37.8} & \textbf{45.4} & 0.838 & 52.7 \\
\bottomrule
\end{tabular}
}
\end{table}

\begin{table}[htbp]
\caption{Effect of the Sinkhorn regularizer $\varepsilon$ on NQ. }
\label{tab:eps_sweep}
\centering
\renewcommand{\arraystretch}{1.3}
\begin{tabular}{lcccc}
\toprule
$\varepsilon$ & 0.001 & 0.01 & 0.05 & 0.1 \\
\midrule
NQ EM     & 54.1 & \textbf{54.8} & 54.3 & 53.2 \\
$F_Q(S)$  & \textbf{0.841} & 0.838 & 0.822 & 0.804 \\
\bottomrule
\end{tabular}
\end{table}

The comparison shows that the entropic Sinkhorn surrogate is not merely a differentiable convenience, but a genuine source of robustness. Although exact-FL greedy optimizes the hard objective directly and attains a marginally higher realized $F_Q(S)$ ($0.842$ vs.\ $0.838$), the deployed Sinkhorn selector \emph{outperforms} it on downstream answer quality ($+0.7$ Avg EM, $45.4$ vs.\ $44.7$), with the gap widening on multi-hop HotpotQA ($+0.7$) and 2WikiMHQA ($+0.8$) relative to single-hop NQ ($+0.4$). The gain thus comes not from optimizing $F_Q$ more tightly---exact greedy wins that contest---but from the entropic smoothing itself: the soft transport plan spreads each unit of demand across several near-tied candidate servers rather than committing to a single nearest chunk, hedging against proxy noise in $\mathcal{P}_Q$, precisely in the multi-modal regime where it matters most. The $\varepsilon$-sweep (Table~\ref{tab:eps_sweep}) makes this explicit: NQ EM peaks at an intermediate $\varepsilon=0.01$ and degrades toward both the hard limit ($\varepsilon\to0$, $54.1$) and the over-smoothed regime ($\varepsilon=0.1$, $53.2$), while $F_Q(S)$ decreases monotonically in $\varepsilon$---so the EM-optimal $\varepsilon$ is \emph{decoupled} from the $F_Q$-optimal one, confirming that entropic regularization adds robustness beyond hard coverage maximization rather than acting as a numerical artifact. We accordingly retain the Sinkhorn realization as the default selector, with the $O(Kk)$ exact-FL greedy available as a fast, OT-free fallback that provably inherits the $(1-1/e)$ guarantee but forgoes this smoothing-induced robustness; the ``$\Delta W_\varepsilon$ only'' row further confirms that the reassignment gate is separately load-bearing ($-6.4$ Avg EM when removed), leaving the gate and coverage term as complementary facets of the marginal gain.

\subsection{Component-wise Contribution Analysis}

Table~\ref{tab:component_ablation} reports a component-wise ablation study on NQ, where each row removes one component from the full GeoRAG while keeping the rest unchanged.

\begin{table}[htbp]
\caption{Component-wise Contribution Ablation on NQ.}
\label{tab:component_ablation}
\centering
\renewcommand{\arraystretch}{1.3}
\resizebox{\columnwidth}{!}{
\begin{tabular}{lcccp{3.5cm}}
\toprule
\textbf{Ablation Variant} & \textbf{NQ EM} & \textbf{NQ F1} & $\Delta$ \textbf{EM}
& \textbf{Description} \\
\midrule
w/o Reassignment Gate            & 48.3 & 60.1 & $-6.5$
& Redundancy-suppression facet of coverage gain \\
w/o $\Delta W$ Coverage Term     & 49.5 & 62.0 & $-5.3$
& Reclaimed-demand-mass facet of coverage gain \\
w/o Two-stage Sub-queries        & 52.4 & 65.1 & $-2.4$
& Greedy max-min diversity filtering \\
w/o Reverse Validation           & 54.1 & 67.2 & $-0.7$
& Quality weighting of sub-queries \\
\textbf{GeoRAG}                  & \textbf{54.8} & \textbf{67.9} & \textbf{---} & \\
\bottomrule
\end{tabular}
}
\end{table}

The reassignment gate ($-6.5$ EM) and $\Delta W$ ($-5.3$ EM) are the two largest contributions and are indispensable. In the exact facility-location objective the two coincide with the single marginal gain $\Delta F_Q$ (Section~\ref{sec:greedy}); empirically, however, each remains separately load-bearing under the Sinkhorn realization, so we treat them as \emph{complementary} rather than interchangeable (Section~\ref{sec:exact_vs_sinkhorn}): removing the gate lets coverage be optimized but allows redundant chunks to occupy slots, while removing $\Delta W$ lets the gate repel similar chunks but fails to guide selection toward demand gaps. Two-stage sub-query generation contributes $-2.4$ EM when removed (replaced by single-stage top-10 selection), showing the importance of max-min filtering for sub-query space coverage. Reverse validation achieves $+0.7$ EM at a latency cost of only $8.3$ ms, a highly cost-effective design.

\subsection{Direct Measurement of Demand-Dimension Coverage}
\label{sec:coverage_metric}

The main results measure the \emph{downstream} effect of coverage (final EM). To measure the mechanism \emph{directly---whether} GeoRAG actually covers more demand dimensions than a single-point baseline---we define a \emph{demand dimension coverage rate} on bimodal multi-hop queries and compare Cosine top-5 with GeoRAG in Figure~\ref{fig:coverage}.

\begin{definition}[Demand-dimension coverage rate]
For a query whose gold evidence is partitioned into two semantically independent dimensions (Peak~A / Peak~B), label each selected chunk by the peak it supports. The \emph{per-peak coverage} is the fraction of queries for which at least one selected chunk supports that peak; the \emph{both-peaks coverage} is the fraction for which both peaks are supported.
\end{definition}

\textbf{Protocol.} We restrict ourselves to HotpotQA and 2WikiMHQA with annotated supporting facts. For each query, we map its two gold supporting passages to Peak~A and Peak~B. Each selected chunk is assigned to the peak whose gold passage it matches the most closely (cosine above threshold $0.6$). We report per-peak and both-peaks coverage with paired bootstrap.

\begin{figure}[htbp]
    \centering
    \includegraphics[width=\linewidth]{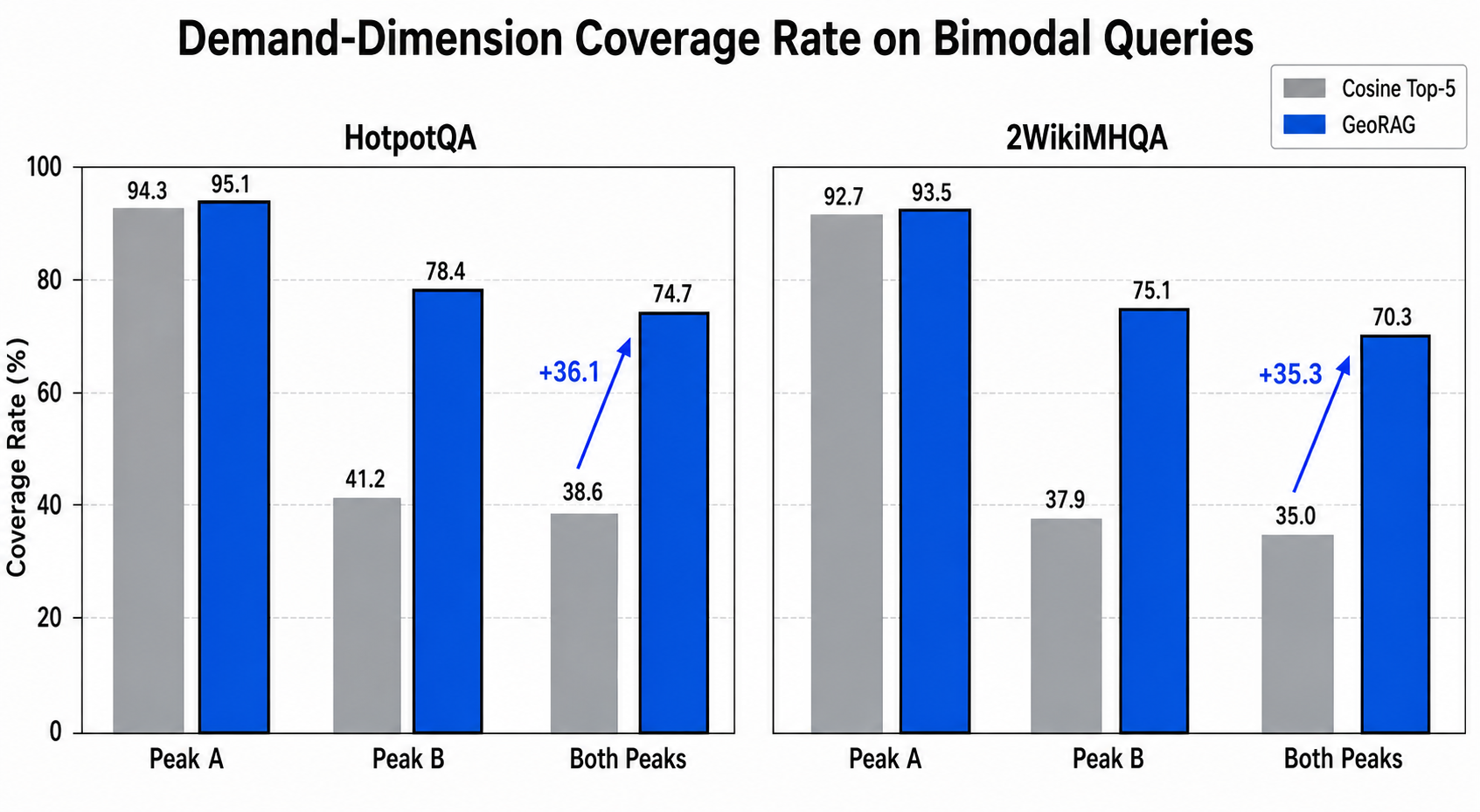}
    \caption{Demand-dimension coverage rate. }
    \label{fig:coverage}
\end{figure}

Both methods saturate Peak~A, but Cosine top-5 covers Peak~B only for $41.2\%$ of HotpotQA queries and $37.9\%$ of 2WikiMHQA, while GeoRAG expands Peak~B coverage to $78.4\%$ and $75.1\%$, respectively. Consequently, both-peaks coverage rises from $38.6\%$ to $74.7\%$ on HotpotQA ($+36.1$) and from $35.0\%$ to $70.3\%$ on 2WikiMHQA ($+35.3$). This coverage gap tracks the downstream EM gap, indicating that the EM improvement is mediated by improved dimensional coverage rather than an unrelated factor.

\subsection{Robustness to the Context Budget $k$}
\label{sec:k_robustness}

Table~\ref{tab:k_robustness} varies $k\in\{3,5,8,10\}$ under the Dense retriever and compares GeoRAG with the Cosine top-$k$ truncation. We also validated the size of the candidate pool: EM increases sharply from $K=10$ to $K=200$ ($+6.1$ EM) and then plateaus ($+0.3$ EM from $K=200$ to $500$ while the sinkhorn latency increases to 137.2 ms, an increase of $2.6\times$, confirming $K=200$ as the Pareto-optimal operating point.

\begin{table}[htbp]
\caption{Robustness to the context budget $k$ (Dense retriever, EM). }
\label{tab:k_robustness}
\centering
\renewcommand{\arraystretch}{1.3}
\resizebox{\columnwidth}{!}{
\begin{tabular}{lcccc}
\toprule
\textbf{Budget} & \textbf{NQ Cosine} & \textbf{NQ GeoRAG} & \textbf{HotpotQA Cosine} & \textbf{HotpotQA GeoRAG} \\
\midrule
$k=3$  & 41.2 & 51.3 \,($+10.1$) & 30.9 & 40.0 \,($+9.1$) \\
$k=5$  & 44.1 & 54.8 \,($+10.7$) & 33.8 & 43.5 \,($+9.7$) \\
$k=8$  & 46.0 & 56.0 \,($+10.0$) & 35.5 & 44.9 \,($+9.4$) \\
$k=10$ & 46.8 & 56.4 \,($+9.6$)  & 36.2 & 45.3 \,($+9.1$) \\
\bottomrule
\end{tabular}
}
\end{table}

The GeoRAG gain is stable across the entire budget range and is slightly larger at small budgets: when only $k=3$ slots are available, the cost of spending all of them on a single demand dimension is highest, so directed coverage matters most. The improvement is not a budget artifact---it reflects a structural difference in \emph{which} chunks are chosen, and persists even when the context window is generous.

\subsection{Generalization Across Sub-query Generators}

Because the sub-query generator is the only learned component touching GeoRAG (used zero-shot, with no fine-tuning), we test whether the method is tied to the specific Qwen3-4B backbone by swapping it for four alternative generators while holding the embedding model, retriever, and selection procedure fixed. The results are in Table~\ref{tab:generator}.

\begin{table}[htbp]
\caption{Generalization across sub-query generators.}
\label{tab:generator}
\centering
\renewcommand{\arraystretch}{1.25}
\resizebox{\columnwidth}{!}{
\begin{tabular}{lcc}
\toprule
\textbf{Sub-query Generator} & \textbf{NQ EM} & \textbf{HotpotQA EM} \\
\midrule
Mistral-7B-Instruct          & 53.9 & 42.4 \\
Llama-3.1-8B-Instruct        & 54.3 & 42.1 \\
\textbf{Qwen3-4B}            & 54.8 & 43.5 \\
Qwen3-8B                     & 55.1 & 43.7 \\
Qwen3.5-4B                   & 55.3 & 44.2 \\
\bottomrule
\end{tabular}
}
\end{table}

All five generators produce EM within a $1.4$-point band on NQ ($2.1$ on HotpotQA) and all outperform every baseline in Table~\ref{tab:baselines}. Stronger generators give a small additional lift, but even an open 4B model preserves most of the gain. GeoRAG's improvement comes from the coverage \emph{formulation} rather than from a privileged backbone; the framework is generator-agnostic and can ride future improvements in sub-query generation without redesign.

\section{Conclusion}
\label{sec:conclusion}

We presented GeoRAG, an unsupervised, training-free, retrieval-agnostic context selection framework that reformulates RAG context selection as Information Demand Coverage Optimization. By building a multi-dimensional Information Demand Proxy $\mathcal{P}_Q$ (Axis~A) and greedily minimizing the Sinkhorn--Wasserstein distance between $\mathcal{P}_Q$ and the coverage measure of the selected set (Axis~B), GeoRAG improves EM by $+6.5$--$+7.5$ over direct top-5 truncation across six open-domain QA benchmarks and six retrieval methods, with the largest gains on multi-hop and ambiguous tasks. We further proved that no query-proximity-monotone selector can cover a bimodal demand regardless of model capacity, and showed empirically that cross-encoder rerankers (BGE-Reranker) inherit the same failure, confirming the limitation is structural rather than incidental. As a drop-in post-retrieval module needing no labeled data or training, GeoRAG integrates with any existing RAG pipeline.

\end{document}